\documentclass[letterpaper, 10pt, conference]{ieeeconf}
\IEEEoverridecommandlockouts
\usepackage{etex}

\usepackage{graphicx}
\usepackage{balance}
\usepackage{cite}
\usepackage{soul}
\usepackage{epstopdf}
\usepackage{amssymb}
\usepackage[cmex10]{amsmath}
\usepackage{cases}
\let\labelindent\relax
\usepackage{enumitem}
\usepackage{booktabs}

\usepackage{array}
\usepackage{multirow}
\usepackage{bigdelim}
\usepackage{mdwtab}
\usepackage{tikz}
\usepackage{physics}
\usepackage{color}

\makeatletter
\let\MYcaption\@makecaption
\makeatother
\usepackage[font=footnotesize]{subcaption}
\makeatletter
\let\@makecaption\MYcaption
\makeatother

\usepackage{bm}

\DeclareMathOperator*{\arginf}{arg\,inf}

\usepackage{accents}

\newtheorem{theorem}{Theorem}
\newtheorem{lemma}{Lemma}

\newtheorem{remark}{Remark}

\title{\bfseries Active Trajectory Estimation for Partially Observed Markov Decision Processes via Conditional Entropy}%
\author{Timothy L.\ Molloy and Girish N.\ Nair%
\thanks{The authors are with the Department of Electrical and Electronic Engineering, University of Melbourne, Parkville, VIC, 3010, Australia.
{\texttt{\{tim.molloy,gnair\}@unimelb.edu.au}}}
\thanks{This work received funding from the Australian Government, via grant AUSMURIB000001 associated with ONR MURI grant N00014-19-1-2571.}%
}

\begin{document}




\maketitle
\thispagestyle{empty}
\pagestyle{empty}

\begin{abstract}
\boldmath
In this paper, we consider the problem of controlling a partially observed Markov decision process (POMDP) in order to actively estimate its state trajectory over a fixed horizon with minimal uncertainty.
We pose a novel active smoothing problem in which the objective is to directly minimise the {\em smoother entropy}, that is, the conditional entropy of the (joint) state trajectory distribution of concern in fixed-interval Bayesian smoothing.
Our formulation contrasts with prior active approaches that minimise the sum of conditional entropies of the (marginal) state estimates provided by Bayesian filters.
By establishing a novel form of the smoother entropy in terms of the POMDP belief (or information) state, we show that our active smoothing problem can be reformulated as a (fully observed) Markov decision process with a value function that is concave in the belief state.
The concavity of the value function is of particular importance since it enables the approximate solution of our active smoothing problem using piecewise-linear function approximations in conjunction with standard POMDP solvers.
We illustrate the approximate solution of our active smoothing problem in simulation and compare its performance to alternative approaches based on minimising marginal state estimate uncertainties.
\end{abstract}

%
\IEEEpeerreviewmaketitle

\section{Introduction}
The problem of active state estimation involves controlling a partially observed stochastic dynamical system in order to elicit useful information for estimating its partially observed state \cite{Krishnamurthy2016, Blackmore2008,Hu2004,Baglietto2007,Scardovi2007,Zois2017}.
Active state estimation has been investigated under a variety of names across a range of applications including controlled sensing and sensor scheduling \cite{Nitinawarat2013, Evans2001, Wu2008, Krishnamurthy2007}, dual control \cite{Bar1974,Mesbah2018,Flayac2017}, fault detection \cite{Esna2012,Tzortzis2019,Heirung2019}, target detection and tracking \cite{Chattopadhyay2018, Zois2017, Zois2014}, active learning \cite{Araya2010, Riccardi2005}, uncertainty-aware robot navigation \cite{Thrun2005,Nardi2019}, and active simultaneous localisation and mapping (SLAM) \cite{Mu2016, Roy2005, Valencia2018, Carrillo2012, Thrun2005, Valencia2018}.
The principal challenge in active state estimation lies in finding meaningful estimation performance measures that are tractable to optimise within standard partially observed stochastic optimal control frameworks such as partially observed Markov decision processes (POMDPs).
Most treatments of active state estimation therefore optimise estimation performance measures that directly relate to the performance of Bayesian filters, since Bayesian filters are inherently used to solve POMDPs.
Bayesian filters estimate the current state at each time instant, given all available measurements and controls until then.
However, in applications including target tracking, active SLAM, and uncertainty-aware navigation, state \emph{trajectory} estimates are of greater interest than marginal \emph{instantaneous} state estimates.
For instance, in surveillance applications, it can be important to estimate not just where a target currently is, but from where it came and what points it visited. In SLAM, better estimates of the past trajectory can also help reconstruct a more accurate map of the environment.
Motivated by such applications, in this paper we investigate a novel active state estimation problem with an estimation performance measure directly related to state trajectory uncertainty.

Bayesian (fixed-interval) smoothinng is concerned with inferring the state of a partially observed stochastic dynamical system given an entire trajectory of measurements and controls.
Unlike Bayesian filters, Bayesian smoothers are thus capable of exploiting past, present, and future measurements and controls to compute state estimates (cf.~\cite{Briers2010}).
Bayesian smoothing has been exhaustively studied over many decades, with smoothing algorithms being key components in many state-of-the-art target tracking systems (cf.~\cite{Bar-Shalom2001}) and robot SLAM systems (cf.~\cite{Thrun2005}).
The problem of controlling a system in order to estimate its state trajectory with smoother-like algorithms has received some (limited) recent attention in the context of active SLAM for robotics (cf.~\cite{Mu2016,Carrillo2012,Thrun2005,Valencia2018}).
However, many fundamental challenges remain including formulating meaningful smoother estimation performance measures that are amiable to optimisation with standard POMDP algorithms (which are increasingly able to handle large state and measurement spaces, see \cite{Krishnamurthy2016}).

Popular estimation performance measures proposed previously for active state estimation have included estimation error probabilities \cite{Krishnamurthy2007,Blackmore2008}, mean–squared error \cite{Zois2014,Zois2017,Krishnamurthy2007}, Fisher information \cite{Flayac2017}, and the (Shannon or R\'enyi) entropy of estimates from Bayesian filters \cite{Krishnamurthy2007,Scardovi2007,Thrun2005} (see \cite[Chapter 8]{Krishnamurthy2016} and references therein for more examples).
Active state estimation with these popular estimation performance measures is typically formulated as a POMDP and solved by reformulating it as a (fully observed) Markov decision process (MDP) in terms of a belief (or information) state (cf.~\cite[Chapter 8]{Krishnamurthy2016}).
The belief state corresponds to Bayesian filter estimates, which makes MDP reformulations straightforward in the case of popular estimation performance measures.
In contrast, estimation performance measures that explicitly relate to Bayesian smoothing and that can be expressed as functions of the belief state appear yet to be considered.

A sizable body of literature has dealt with the solution of POMDPs that specifically arise in active state estimation with popular estimation performance measures.
In general, the solution of these POMDPs is intractable, however, in some important cases, theoretical results have provided useful insight into the structure and nature of their solutions (see \cite{Krishnamurthy2007,Krishnamurthy2016,Krishnamurthy2020,Araya2010} and references therein).
This theoretical insight has enabled the construction of arbitrary-error approximate solutions using standard algorithms for solving POMDPs (cf. \cite{Araya2010}) and the construction of myopic policies that bound the optimal policy under certain dominance conditions (cf. \cite[Chapter 14]{Krishnamurthy2016}).
In addition to proposing an active trajectory estimation problem, we shall also seek to establish theoretical results characterising the structure of its solutions with the aim of identifying tractable approximate solutions.

The main contribution of this paper is the proposal of a novel active smoothing problem in which a POMDP is controlled to reduce the uncertainty associated with its state trajectory.
In contrast to prior treatments of active state estimation, we directly minimise the uncertainty of state trajectory estimates provided by (fixed-interval) Bayesian smoothers rather than upper bounds on state trajectory uncertainty based on state estimates from Bayesian filters.
An important secondary contribution of this paper is the reformulation of our active smoothing problem as a fully observed MDP with a concave value function in terms of the standard concept of belief (or information) state for POMDPs.
Our belief-state reformulation and novel concavity result enables the approximate solution of our active smoothing problem using standard POMDP solution methods and piecewise-linear function approximations.
We illustrate the approximate solution of our active smoothing problem in simulations where its performance to standard active state estimation approaches is also examined.

This paper is structured as follows.
In Section \ref{sec:problem} we pose our active smoothing problem.
In Section \ref{sec:active} we construct a belief-state reformulation of our active smoothing problem, present dynamic programming equations and structural results for solving it.
Finally, we illustrate and compare our active smoothing problem with other approaches in Section \ref{sec:results} and present conclusions in Section \ref{sec:conclusion}.

\emph{Notation:} We denote random variables with capital letters such as $X$, and their realisations with lower case letters such as $x$.
We assume all random variables have probability mass functions (or densities when they are continuous), with the probability mass function of $X$ written as $p(x)$, the joint probability mass function of $X$ and $Y$ written as $p(x, y)$, and the conditional probability mass function of $X$ given $Y = y$ written as $p(x|y)$ or $p(x | Y = y)$.
For a function $f$ of $X$, the expectation of $f$ evaluated with $p(x)$ will be denoted $E_X [f(x)]$ and the conditional expectation of $f$ evaluated with $p(x|y)$ as $E[f(x) | y]$.
The {\em point-wise} entropy of $X$ given $Y = y$ will be written $H(X | y) \triangleq - \sum_{x} p(x|y) \log p(x|y)$ with the (average) conditional entropy of $X$ given $Y$ being $H(X | Y) \triangleq E_{Y} \left[ h(X|y) \right]$.
The mutual information between $X$ and $Y$ is $I(X; Y) \triangleq H(X) - H(X | Y) = H(Y) - H(Y | X)$. The point-wise conditional mutual information of $X$ and $Y$ given $Z = z$ is $I(X; Y | z) \triangleq H(X | z) - H(X | Y, z)$ with the (average) conditional mutual information given by $I(X; Y | Z) \triangleq E_{Z} \left[ I(X; Y | z) \right]$.
Where there is no risk of confusion, we will occasionally omit the adjectives ``point-wise'' and ``conditional''.

\section{Problem Formulation and Approach}
\label{sec:problem}
In this section, we pose our active smoothing problem and sketch is solution as a POMDP.

\subsection{Problem Formulation}
Let $X_k$ for $k \geq 0$ be a discrete-time first-order Markov chain with discrete finite state-space $\mathcal{X} \triangleq \{e_1, \ldots, e_N\}$ where $e_i$ is an indicator vector of appropriate dimensions with $1$ in its $i$th component and zeros elsewhere.
We shall denote the initial probability distribution of $X_0$ as $\pi_0 \in \mathbb{R}^N$ with $i$th component $\pi_0(i) \triangleq P(X_0 = e_i)$, and we shall let the (controlled) transition dynamics of $X_k$ be described by the state transition probabilities:
\begin{align}
    \label{eq:stateProcess}
    A^{ij}(u) \triangleq p( X_{k+1} = e_i | X_k = e_j, U_k = u)
\end{align}
with the controls $u$ from the process $U_k$ belonging to a discrete finite set $\mathcal{U}$.
The state process $X_k$ is (partially) observed through a stochastic measurement process $Y_k$ for $k \geq 0$ taking values in some (potentially discrete) metric space $\mathcal{Y}$.
The measurements $Y_k$ are conditionally independent given the states $X_k$, and are distributed according to the measurement kernels:
\begin{align}
    \label{eq:obsProcess}
    B^{i} (Y_k, u) \triangleq p( Y_k | X_k = e_i, U_{k-1} = u)
\end{align}
for $k \geq 1$ with $B^{i}(Y_0) \triangleq p( Y_0 | X_0 = e_i)$. We note that the measurement kernels $B$ will constitute conditional probability density functions when $\mathcal{Y}$ is continuous and conditional probability mass functions when $\mathcal{Y}$ is finite and discrete.
The tuple $\lambda \triangleq (\pi_0, A, B)$ is a controlled hidden Markov model (HMM) \cite{Elliott1995}.

The controlled HMM $\lambda$ constitutes a standard POMDP when the controls $U_k$ are given by a (potentially stochastic) output feedback control policy $\mu$ that solves
\begin{align}
 \label{eq:standardPOMDP}
 \begin{aligned}
   \inf_\mu E_\mu \left[ c_T(x_T) + \sum_{k = 0}^{T-1} c_k \left(x_k, u_k\right) \right]
 \end{aligned}
\end{align}
subject to the state and measurement processes \eqref{eq:stateProcess} and \eqref{eq:obsProcess} for a given horizon $T < \infty$ (cf.\ \cite[Section 7.1]{Krishnamurthy2016}).
Here, the control policy $\mu \triangleq \left\{ \mu_{k} : 0 \leq k < T \right\}$ is defined by conditional probability kernels $\mu_{k}(y^k, u^{k-1}) \triangleq p(u_{k} | y^k, u^{k-1})$ given measurements and controls $y^k \triangleq \{y_0, \ldots, y_k\}$ and $u^{k-1} \triangleq \{u_0, \ldots, u_{k-1}\}$, and the expectation $E_\mu [\cdot]$ is over the joint distribution of the states $X^{T}$ and measurements $Y^{T} \triangleq \{Y_0, \ldots, Y_T\}$ (with the policy $\mu$ then implying a distribution on the controls $U^{T-1}$).
Furthermore, the functions $c_T : \mathcal{X} \mapsto [0, \infty)$ and $c_k : \mathcal{X} \times \mathcal{U} \mapsto [0, \infty)$ are terminal and instantaneous cost functions that encode desired system performance such as reducing control effort or penalising deviations of the state from a desired trajectory.

Standard Bayesian (fixed-interval) smoothing is concerned with estimating the states $X^T \triangleq \{X_0, \ldots, X_T\}$ given measurement and control realisations $y^T$ and $u^{T-1}$.
Whilst the provenance of the controls $u^{T-1}$ is not an explicit concern in standard Bayesian smoothing, in general, the controls have the potential to affect both the state values and the uncertainty associated with them in a phenomenon known as the dual control effect \cite{Bar1974}.
In order to exploit this effect, let us quantifying the estimation performance of Bayesian smoothers using the conditional joint entropy
\begin{align}
    \label{eq:condEntCriteria}
    H(X^T | Y^T, U^{T-1}) = E_{Y^T,U^{T-1}}[H(X^T | y^T, u^{T-1})]
\end{align}
where $H(X^T | y^T, u^{T-1})$ is the point-wise conditional entropy of the (joint) smoother estimate (i.e., the entropy of the joint conditional distribution $p(x^T | y^T, u^{T-1})$ over state realisations $x^T \triangleq \{x_0, \ldots, x_T\}$).
Our active smoothing problem is then to find a policy $\mu$ that minimises the smoother entropy and the costs $c_k$ and $c_T$ by solving
\begin{align}
 \label{eq:activeSmoothing}
 \begin{split}
   \inf_\mu \Bigg\{ &H(X^T | Y^T, U^{T-1}) \\
   &\quad+ E_\mu \left[ c_T(x_T) + \sum_{k = 0}^{T-1} c_k \left(x_k, u_k\right) \right] \Bigg\}
 \end{split}
\end{align}
subject to \eqref{eq:stateProcess} and \eqref{eq:obsProcess}.

Our proposal of the smoother entropy $H(X^T | Y^T, U^{T-1})$ as a measure of estimation performance is primarily motivated by the interpretation of entropy as a measure of uncertainty --- the smaller the smoother entropy, the more concentrated we expect the smoother distribution $p(x^T | y^T, u^{T-1})$ to be at the true (unknown) state sequence realisation $X^T = x^T$.
Our proposal also contrasts with previous approaches that have used the sum of uncertainty in marginal (or instantaneous) estimates of the state $X_k$.
For example, the point-wise conditional entropy $H(X_k | y^k, u^{k-1})$ has frequently been added to the terminal and/or instantaneous costs $c_T$ and $c_k$ in active sensing and robotics (see \cite{Krishnamurthy2007,Krishnamurthy2016,Araya2010,Thrun2005, Mu2016} and references therein for details).
However, this does not directly account for correlations between subsequent states, and thus overestimates the uncertainty in the trajectory.
Indeed, the expected sum of such entropy terms is strictly greater than the smoother entropy since 
\begin{align}
\begin{split}
\label{eq:entropyBounds}
\sum_{k = 0}^T H(X_k | Y^{k}, U^{k-1})
&\geq  H(X^T | Y^T, U^{T-1}),
\end{split}
\end{align}
with equality holding only when the states are (temporally) independent.
Unlike previous approaches, our approach \eqref{eq:activeSmoothing} therefore explicitly encourages exploitation of the temporal dependencies between states, and hence directly aids state estimators that use the entire trajectory of measurements $Y^T$ and controls $U^{T-1}$ such as Bayesian smoothers \cite{Briers2010,Elliott1995} and the Viterbi algorithm \cite{Rabiner1989}.

\subsection{POMDP Solution Approach}
Whilst our active smoothing problem \eqref{eq:activeSmoothing} constitutes a POMDP, its solution in the same manner as standard POMDPs of the form in \eqref{eq:standardPOMDP} is complicated by the smoother entropy $H(X^T | Y^T, U^{T-1})$.
Firstly, the smoother entropy $H(X^T | Y^T, U^{T-1})$ does not constitute a standard terminal or instantaneous cost.
Furthermore, the solution (or approximate solution) of standard POMDPs \eqref{eq:standardPOMDP} involves reformulating them as (fully observed) MDP in terms of a belief (or information) state corresponding to the state estimates $p(X_k | y^k, u^{k-1})$ given by Bayesian filters.
Naive belief-state reformulations of the smoother entropy $H(X^T | Y^T, U^{T-1})$ however lead only to the upper bounds in \eqref{eq:entropyBounds}.
Finally, existing algorithms for solving POMDPs (or finding tractable approximate solutions) require certain structural properties of the cost and value functions, including concavity.

In this paper, we focus on establishing a novel form of the smoother entropy $H(X^T | Y^T, U^{T-1})$. This will let us reformulate our active smoothing problem as an MDP with structural properties amenable to the use of standard POMDP algorithms for finding tractable (approximate) solutions.


\section{Belief-State Reformulation, Structural Results, and Approximate Solution Approach}
\label{sec:active}

In this section, we introduce a belief-state reformulation of our active smoothing problem by establishing a novel form of the smoother entropy $H(X^T | Y^T, U^{T-1})$.
We use this reformulation to derive key structural results and an approximate solution approach.

\subsection{Belief-State Reformulation}

As a first step towards reformulating our active smoothing problem, we shall establish a novel additive form of the smoother entropy $H(X^T | Y^{T-1}, U^{T-1})$ for the HMM $\lambda$.

\begin{lemma}
\label{lemma:stageAdditive}
 The smoother entropy $H(X^T | Y^T, U^{T-1})$ of the controlled HMM $\lambda$ with controls $U_k$ given by some (potentially stochastic) output-feedback policy $\mu = \{\mu_k : 0 \leq k < T\}$ with $\mu_{k}(y^k, u^{k-1}) = p(u_{k} | y^k, u^{k-1})$ has the additive form:
\begin{align}
    \label{eq:stageAdditive}
    \begin{split}
    H(X^T | Y^T, U^{T-1}) 
    &= E_\mu \Bigg[ H(X_T | y^{T}, u^{T-1})\\
    &\;+ \sum_{k = 0}^{T-1} H(X_{k} | X_{k+1}, y^{k}, u^{k}) \Bigg]
    \end{split}
\end{align}
where we define $H(X_0 | Y^{0}, U^{-1}) \triangleq H(X_0 | Y_0)$.
\end{lemma}
\begin{proof}
Proved via induction on $T$ in a similar manner to \cite[Lemma 1]{Molloy2021}.
\end{proof}

\begin{remark}
A different additive form for $H(X^T | Y^T, U^{T-1})$ was established recently for general nonlinear state-space models in our manuscript \cite{Molloy2021}, which explores the opposite goal of \emph{maximising} the smoother entropy.
In particular, the form in \cite{Molloy2021} involves the subtraction of entropy terms, whilst \eqref{eq:stageAdditive} involves only addition.
We will later see that this difference is important in establishing the structural properties of our active smoothing (minimisation) problem, since addition preserves concavity.
\end{remark}

\begin{remark}
The terms $H(X_k | X_{k+1}, Y^k, U^k)$ that appear after taking the expectation in \eqref{eq:stageAdditive} can be rewritten as the difference $H(X_k | X_{k+1}, Y^k, U^k) = H(X_k | Y^k, U^{k-1}) - I(X_k;X_{k+1}|Y^k, U^k)$ so that \eqref{eq:stageAdditive} becomes,
\begin{align*}
    &H(X^T | Y^T, U^{T-1})\\
    &\quad= \sum_{k = 0}^T H(X_k | Y^{k}, U^{k-1}) - \sum_{k = 0}^{T-1} I(X_k;X_{k+1}|Y^k, U^k).
\end{align*}
In this form, we see that previous approaches that minimise the sum of (marginal) state entropies $H(X_k | Y^{k}, U^{k-1})$ (as described above in \eqref{eq:entropyBounds}) neglect the potential for the smoother entropy $H(X^T | Y^T, U^{T-1})$ to be reduced by increasing the conditional mutual information $I(X_k;X_{k+1}|Y^k, U^k)$ (or dependency) between consecutive states.
\end{remark}

Lemma \ref{lemma:stageAdditive} is of considerable practical value since the entropies $H(X_T | y^{T}, u^{T-1})$ and $H(X_{k} | X_{k+1}, y^{k}, u^{k})$ in \eqref{eq:stageAdditive} have straightforward belief-state reformulations.
Specifically, let us define the belief (or information) state $\pi_{k} \in \Delta^N$ as the distribution of the state $X_k$ given previous measurements and controls, and with $i$th component $\pi_{k}(i) \triangleq p(X_{k} = e_i | y^k, u^{k-1})$.
The belief state belongs to the probability simplex $\Delta^N \triangleq \{x \in \mathbb{R}^n : \sum_{i = 1}^N x(i) = 1, \; 0 \leq x(i) \leq 1 \; \forall i\}$.
Let us also define the joint predicted belief state $\bar{\pi}_{k+1 | k} \in \Delta^{N^2}$ as the joint probability distribution of the states $X_{k+1}$ and $X_{k}$ with $(i,j)$ component $\bar{\pi}_{k+1 | k} (i, j) \triangleq p(X_{k+1} = e_i, X_{k} = e_j | y^{k}, u^{k})$.
The belief state $\pi_{k+1}$ and joint predicted belief state $\bar{\pi}_{k+1 | k}$ are related via the Bayesian filter prediction step,
\begin{align}
    \label{eq:bayesianPred}
    \bar{\pi}_{k+1 | k}(i,j)
    &= A^{ij}(u_{k}) \pi_{k}(j)
\end{align}
for all $1 \leq i,j \leq N$ and $0 \leq k \leq T-1$, and the Bayesian filter update step
\begin{align}
    \label{eq:bayesTemp}
    \pi_{k+1}(i)
    &= \dfrac{ B^i(y_{k+1}, u_{k}) \sum_{j = 1}^N \bar{\pi}_{k+1 | k} (i,j)}{\sum_{\ell = 1}^N\sum_{j = 1}^N B^\ell(y_{k+1},u_{k}) \bar{\pi}_{k+1 | k} (\ell,j)}
\end{align}
for all $1 \leq i \leq N$ and $0 \leq k \leq T-1$, with initial (prior) belief $\pi_0$.
We shall use $\Pi$ to denote the mapping defined by the successive application of the Bayesian filter prediction and update steps \eqref{eq:bayesianPred} and \eqref{eq:bayesTemp} to a belief state $\pi_{k}$ with control and measurement $u_{k}$ and $y_{k+1}$, namely,
\begin{align}
    \label{eq:bayesianFilter}
    \pi_{k+1}
    &= \Pi(\pi_{k}, u_{k}, y_{k+1}).
\end{align}

The entropy $H(X_T | y^{T}, u^{T-1})$ in \eqref{eq:stageAdditive} is the entropy of the terminal belief state $\pi_T$, thus we write
\begin{align}\label{eq:terminal_cost_entropy}
    H(X_T | y^T, u^{T-1})
    \triangleq \tilde{g}_T (\pi_T).
\end{align}
Similarly, given the joint predicted belief state $\bar{\pi}_{k+1 | k}$ and the prediction relationship \eqref{eq:bayesianPred}, we can see that the conditional entropy $H(X_{k} | X_{k+1}, y^{k}, u^{k})$ appearing in \eqref{eq:stageAdditive} can be viewed as a function of $\pi_k$ and $u_k$ in the sense that,
\begin{align}\notag
    &H(X_{k} | X_{k+1}, y^{k}, u^{k})\\\notag
    &= - \sum_{i,j = 1}^N A^{ij}(u_{k}) \pi_{k}(j) \log \dfrac{A^{ij}(u_{k}) \pi_{k}(j)}{ \sum_{\ell = 1}^N A^{i\ell}(u_{k}) \pi_{k}(\ell)}\\\label{eq:running_cost_entropy}
    &\triangleq \tilde{g}_k(\pi_k, u_k).
\end{align}
Given these belief-state expressions and Lemma \ref{lemma:stageAdditive}, our main reformulation result is that the active smoothing problem \eqref{eq:activeSmoothing} can be expressed as an MDP in terms of the belief state $\pi_k$.

\begin{theorem}
\label{theorem:ocp}
Define the functions
\begin{align*}
    g_k(\pi_{k}, u_{k})
    &\triangleq E_{X_{k}} \left[ \tilde{g}_k(\pi_k, u_k) + c_k (x_{k}, u_{k}) | \pi_{k}, u_{k} \right]
\end{align*}
for $0 \leq k < T$ and 
$
    g_T(\pi_{T})
    \triangleq E_{X_{T}} \left[ \tilde{g}_T(\pi_T) + c_T (x_{T}) | \pi_{T} \right].
$
Then, the active smoothing problem \eqref{eq:activeSmoothing} is equivalent to the MDP (or fully observed stochastic optimal control problem):
\begin{align}
\label{eq:ocp}
\begin{aligned}
&\inf & & E_{Y^{T}} \left[ \left. g_T(\pi_T) + \sum_{k = 0}^{T-1} g_k \left( \pi_{k}, u_{k} \right) \right| \pi_0 \right]\\ 
&\mathrm{s.t.} & &  \pi_{k+1} = \Pi\left( \pi_{k}, u_{k}, y_{k+1} \right)\\
& & & Y_{k+1} \sim p(y_{k+1} | \pi_{k}, u_{k})\\
& & & \mathcal{U} \ni U_k \sim \bar{\mu}_k(\pi_k),
\end{aligned}
\end{align}
with the optimisation being over policies $\bar{\mu} \triangleq \{\bar{\mu}_k : 0 \leq k < T\}$ that are functions of $\pi_k$, i.e., $\bar{\mu}_k(\pi_k) \triangleq p(u_{k} | \pi_k)$.
\end{theorem}
\begin{proof}
Proved in a similar manner to \cite[Theorem 1]{Molloy2021} using Lemma \ref{lemma:stageAdditive}.
\end{proof}

\subsection{Dynamic Programming Equations}

Given the belief-state formulation of our active smoothing problem established in Theorem \ref{theorem:ocp}, we may limit our consideration of optimal polices to deterministic policies $\bar{\mu}$ of the belief state $\pi_k$ in the sense that $u_{k} = \bar{\mu}_k(\pi_k)$ since such optimal policies exist for finite-horizon stochastic optimal control problems (cf.~\cite{Bertsekas2005, Krishnamurthy2016}).
The value function of our active smoothing problem is then defined as
\begin{align*}
    J_k(\pi_k)
    \triangleq \inf E_{Y_{k+1}^{T}} \left[ \left. g_T(\pi_T) + \sum_{\ell = k}^{T-1} g_k \left( \pi_{\ell}, u_{\ell} \right) \right| \pi_{k} \right]
\end{align*}
for $0 \leq k < T$ and $J_T(\pi_T) \triangleq g_T(\pi_T)$ where $Y_{k}^{T} \triangleq \{ Y_k, \ldots, Y_T \}$ and the optimisation is subject to the constraints in \eqref{eq:ocp}.
The value function then satisfies the dynamic programming recursions
\begin{align*}
    \begin{split}
    J_k(\pi_k)
    &= \inf_{u_{k} \in \mathcal{U}} \left\{ g_k \left( \pi_k, u_{k} \right) \right. \\
    &\qquad \left. + E_{Y_{k+1}} \left[ J_{k+1}(\Pi(\pi_{k}, u_{k}, y_{k+1})) | \pi_k, u_{k} \right] \right\}
    \end{split}
\end{align*}
for $0 \leq k < T$, and the optimal policy is given by
\begin{align*}
    \begin{split}
    \bar{\mu}_k^*(\pi_k) 
    &= u_{k}^*
    \in \arginf_{u_{k} \in \mathcal{U}} \left\{ g_k \left( \pi_{k}, u_{k} \right) \right. \\
    &\; \qquad \left. + E_{Y_{k+1}} \left[ J_{k+1}(\Pi(\pi_{k}, u_{k}, y_{k+1})) | \pi_k, u_{k} \right] \right\}.
    \end{split}
\end{align*}

The dynamic programming recursions are, in general, intractable.
However, by exploiting Lemma \ref{lemma:stageAdditive} and \eqref{eq:terminal_cost_entropy} and \eqref{eq:running_cost_entropy}, we will next characterise the structure of the cost and value functions.
These structural results enable the use standard POMDP techniques to construct tractable approximate solutions to our active smoothing problem.

\subsection{Structural Results}

Our first structural result establishes the concavity of the instantaneous and terminal cost functions $g_k(\pi_k, u_k)$ and $g_T(\pi_T)$ in our active smoothing problem \eqref{eq:ocp}.

\begin{lemma}
\label{lemma:concave}
For any control $u_k \in \mathcal{U}$, the instantaneous and terminal costs $g_k(\pi_k, u_k)$ and $g_T(\pi_T)$ in \eqref{eq:ocp} are concave and continuous in the belief state $\pi_k$ for $0 \leq k \leq T$.
\end{lemma}
\begin{proof}
We sketch the proof.
Note first that $g_T$ is the sum of $H(X_T | y^T, u^{T-1})$ (which is concave and continuous in $\pi_T$, cf.\ \cite[Theorem 2.7.3]{Cover2006}) and $E_{X_{T}} \left[ c_T (x_{T}) | \pi_{T} \right]$ (which is linear in $\pi_T$). $g_T$ is thus concave and continuous in $\pi_T$.

Similarly, for any $u_k \in \mathcal{U}$, $g_k$ is the sum of a concave function of $\pi_k$ and a linear function of $\pi_k$ since $H(X_k | X_{k+1}, y^k, u^k)$ is continuous and concave in the joint predicted belief $\bar{\pi}_{k | k-1}$ (cf.\ \cite[Appendix A]{Globerson2007}), which itself is a linear function of $\pi_k$ (cf.\ \eqref{eq:bayesianPred}). $g_k$ is thus concave and continuous in $\pi_T$.
\end{proof}

The concavity of the instantaneous costs $g_k(\pi_k, u_k)$ established in Lemma \ref{lemma:concave} is nontrivial because they involve conditional entropies, which are in general only concave in the joint distribution of their arguments, rather than in the marginal (belief state) distribution we consider (cf.\ \cite[Appendix A]{Globerson2007}).
Our second structural result uses Lemma \ref{lemma:concave} to establish the concavity of our active smoothing problem's value function.

\begin{theorem}
\label{theorem:concave}
 The value function $J_k(\pi_k)$ of our active smoothing problem is concave in $\pi_k$ for $0 \leq k \leq T$.
\end{theorem}
\begin{proof}
 From \cite[Theorem 8.4.1]{Krishnamurthy2016} via Lemma \ref{lemma:concave}.
\end{proof}

The value function $J_k$ of standard POMDPs of the form in \eqref{eq:standardPOMDP} is concave in the belief state $\pi_k$.
This concavity property is fundamental to the operation of classical and modern algorithms for solving standard POMDPs, since the expectation, sum, and inf operators used in them preserve concavity \cite{Smallwood1973, Araya2010, Krishnamurthy2016}.
The importance of Lemma \ref{lemma:concave} and Theorem \ref{theorem:concave} is thus that our active smoothing problem preserves the important concavity properties of standard POMDPs.
As we shall discuss next, Lemma \ref{lemma:concave} in particular opens the possibility of using existing algorithms to find approximate solutions to our active smoothing problem.

\subsection{Approximate Solution Approach}
\label{subsec:approx}
POMDPs with instantaneous and terminal belief-state cost functions $g_k$ and $g_T$ that are piecewise linear (as in the case of standard POMDPs \eqref{eq:standardPOMDP} with only $c_k$ and $c_T$) have value functions that admit finite dimensional representations when the measurement space $\mathcal{Y}$ is discrete or discretised (cf.~\cite{Araya2010} and \cite[Chapter 8]{Krishnamurthy2016}).
These finite dimensional representations imply that the value function $J_k$ can be written in terms of a finite set of belief vectors $\Gamma_k \subset \Delta^N$, namely,
\begin{align}
    \label{eq:vectors}
    J_k (\pi_k) = \min_{\alpha \in \Gamma_k} \left< \pi_k, \alpha \right>
\end{align}
where $\left< \cdot, \cdot \right>$ denotes the inner product.
The finite dimensional representations of the value function also enable algorithms for solving standard POMDPs to efficiently operate on the sets of vectors $\Gamma_k$ (see \cite{Araya2010, Krishnamurthy2016} for details).

Due to the costs $g_k$ and $g_T$ in our active smoothing problem being nonlinear in the belief state, our active smoothing problem lacks a value function that admits an exact finite dimensional representation (even in the case of a discrete measurement space).
However, we are still able to find an approximate solution using existing POMDP algorithms by following the piecewise linear approximation approach proposed in \cite[Section 4]{Araya2010}.
Specifically, consider a finite set $\Xi \subset \Delta^N$ of (arbitrarily chosen) \emph{base points} $\xi$ from the belief simplex $\Delta^N$.
For each control $u \in \mathcal{U}$, let us define the tangent hyperplanes to $g_k(\cdot, u)$ at each $\xi \in \Xi$ as
\begin{align*}
    \omega_\xi^u (\pi)
    \triangleq g_k(\xi, u) + \left< (\pi - \xi), \nabla_\xi g_k(\xi, u) \right>
    = \left< \pi, \alpha_\xi^u \right>
\end{align*}
for any $\pi \in \Delta^N$.
Here, $\nabla_\xi g_k(\xi, u)$ denotes the gradient vector of $g_k(\xi, u)$ with respect to its first argument and $\alpha_\xi^u \in \mathbb{R}^N$ are the vectors $\alpha_\xi^u \triangleq g_k(\xi, u) + \nabla_\xi g_k(\xi, u) - \left< \xi, \nabla_\xi g_k(\xi, u) \right>$.
Since from Lemma \ref{lemma:concave} we have that the costs $g_k$ are concave for each control, they are upper bounded by the tangent hyperplanes $\omega_\xi^u$.
The hyperplanes $\omega_\xi^u$ thus form a piecewise linear (upper bound) approximation to $g_k$, i.e.,
\begin{align}
    \label{eq:pwl}
    g_k(\pi, u)
    \leq \min_{\xi \in \Xi} \omega_\xi^u (\pi) = \min_{\xi \in \Xi} \left< \pi, \alpha_\xi^u \right>.
\end{align}
A piecewise linear (upper bound) approximation to $g_T$ can be constructed in an analogous manner.

Given the piecewise linear approximations of the cost functions $g_k$ and $g_T$ in our active smoothing problem, we can employ existing algorithms for solving POMDPs as described in \cite[Section 3.3]{Araya2010} and \cite[Chapter 7]{Krishnamurthy2016}.
The resulting approximate value function will have a finite dimension representation.
The error between the true and approximate value functions can also be made arbitrarily small by increasing the density of base points and can be bounded (see \cite[Section 4]{Araya2010} for a bound in the infinite-horizon case that can be adapted to a finite-horizon bound by virtue of $T$ being finite).

\section{Illustrative Example}
\label{sec:results}
In this section, we illustrate our active smoothing problem and compare it with alternative active estimation approaches. 

\subsection{Example Set-up}

\begin{figure}[t!]
    \centering
    \includegraphics[width=0.7\columnwidth]{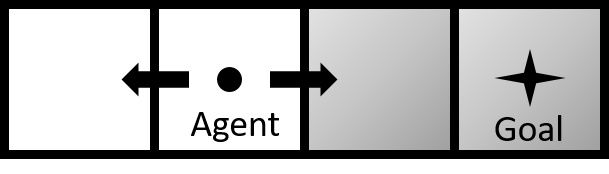}
    \caption{Diagram of the agent considered in our illustrative example that can move west, east, or stay stationary. The agent aims to get close to the goal whilst estimating its trajectory (i.e., path). Measurements are received with different probabilities in the shaded cells versus the unshaded cells.}
    \label{fig:agent}
\end{figure}

We consider an example inspired by uncertainty-aware navigation.
Consider an agent moving in a grid with 4 cells as illustrated in Fig.~\ref{fig:agent}.
Each cell constitutes a state in the agent's state space $\mathcal{X} = \{e_1, \ldots, e_4\}$ (enumerated left to right or west to east).
The agent has three control inputs $\mathcal{U} = \{0,1,2\}$, corresponding to transitioning to the neighbouring cell to the west with probability $0.8$ or staying put with probability $0.2$; staying put with probability $1$; and, transitioning to the adjacent eastern cell with probability $0.8$ or staying put with probability $0.2$, respectively. If a transition would take the agent out of the grid then it remains stationary.
The agent receives two possible measurements $\mathcal{Y} = \{0,1\}$. The agent receives measurement $y = 0$ with probability $0.8$ and measurement $y = 1$ with probability $0.2$ when it is in the two west-most cells, and \emph{vice versa} when it is in the two east-most cells.
Initially, the agent is placed (uniformly) randomly in one of the cells, and over a time horizon of $T = 3$, seeks to move so that it finishes close to the east-most cell with knowledge of the path it took.
We model this situation by considering our active smoothing problem \eqref{eq:activeSmoothing} with $c_k(x_k, u_k) = 0$ and $c_T(e_4) = 1$ but $c_T(e_i) = 0$ for all other $e_i \in \mathcal{X}$.

For the purpose of simulations, we solved our active smoothing problem using the approach detailed in Section \ref{subsec:approx} with a standard POMDP solver\footnote{https://www.pomdp.org/code/} that implements the incremental pruning algorithm.
We selected the base points $\Xi$ in our piecewise linear approximation by constructing a grid with between $1$ and $5$ points linearly spaced in each dimension of the belief from $0$ to $1$. The cost of our active smoothing policy approximated with different numbers of base points per dimension is shown in Fig.~\ref{fig:pwl_study}.
We note that the cost ceases to decrease much after $4$ base points per dimension (we use $5$ in our subsequent results).

For the purpose of comparison, we also found a \emph{Minimum Total Belief Entropy} policy using our approximate solution approach with the same number of base points but with the smoother entropy replaced by the sum of the entropy of each belief state $\pi_k$ over the horizon (i.e., the sum on the left hand side of \eqref{eq:entropyBounds}).
This \emph{Minimum Total Belief Entropy} policy corresponds to previous active state estimation approaches that minimise the entropy of Bayesian filter estimates.
In this example, the complexity of computing (and evaluating) our active smoothing policy is less than that associated with the \emph{Minimum Total Belief Entropy} policy as evidenced by the cardinality of the sets $\Gamma_k$ representing the policies (cf.\ \eqref{eq:vectors}).
Specifically, $|\Gamma_0| = 158$, $|\Gamma_1| = 93$, and $|\Gamma_2| = 46$ for our active smoothing policy compared to $|\Gamma_0| = 438$, $|\Gamma_1| = 224$, and $|\Gamma_2| = 46$ for the \emph{Minimum Total Belief Entropy} policy.

\begin{figure}[t!]
    \centering
    \includegraphics[width = \columnwidth]{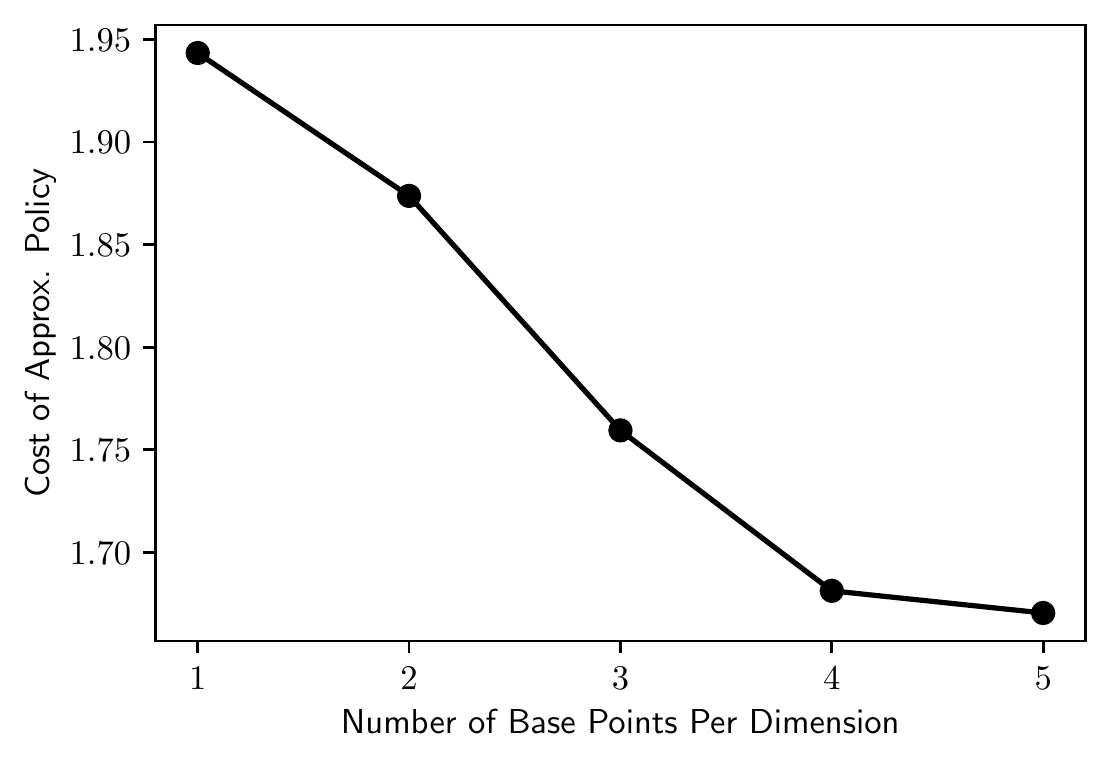}
    \caption{Active smoothing cost \eqref{eq:activeSmoothing} of our approximate active smoothing policy for different numbers of base points (per dimension of the belief state) used in the piecewise linear approximation.}
    \label{fig:pwl_study}
\end{figure}


\begin{table*}[t!]
\begin{center}
\caption{Summary performance measures computed via 10000 Monte Carlo simulations of each policy.}
\label{tbl:illustrativeExample}
\begin{tabular}{@{}lcccc@{}}
\toprule
\multicolumn{1}{c}{\multirow{2}{*}{\textbf{Policy}}} & \textbf{Terminal Cost} & \textbf{Total Belief Entropy} & \textbf{Smoother Entropy} & \multirow{2}{*}{\textbf{\begin{tabular}[c]{@{}c@{}}Total Cost\\ \eqref{eq:activeSmoothing} \end{tabular}}} \\
\multicolumn{1}{c}{}            &                     $E[c_T(x_T)]$ & $\sum_{k = 0}^T H(X_k|Y^k, U^{k-1})$   & $H(X^T | Y^T, U^{T-1})$   &                                                                                \\ \cmidrule(rl){1-5} 
\textbf{Proposed Active Smoothing}                      & 0.5227 & 2.6895                                   & \textbf{1.1518}           & \textbf{1.6745}                                                                \\
\textbf{Minimum Total Belief Entropy}                     & 0.5025 & \textbf{1.9641}                         & 1.5428                    & 2.0453                                                                         \\
\textbf{Always East}                                 & \textbf{0.1495} & 2.3148                                 & 1.7948                    & 1.9443                                                                         \\ \bottomrule
\end{tabular}
\end{center}
\end{table*}

\subsection{Simulation Results}
We performed $10000$ Monte Carlo simulations each for three policies: our active smoothing policy; the \emph{Minimum Total Belief Entropy} policy; and an \emph{Always East} policy that seeks only to reach the goal without estimating the path taken by always selecting the action to move east (i.e., the solution to \eqref{eq:standardPOMDP}, or equivalently, \eqref{eq:activeSmoothing} without the smoother entropy).
Table \ref{tbl:illustrativeExample} summarises the (average) terminal cost $c_T$, total belief entropy, smoother entropies, and total active smoothing cost \eqref{eq:activeSmoothing} computed from the simulations for each of the three policies.
Representative realisations of the controls selected using each of the policies are also shown in Fig.~\ref{fig:actions} (with the agent starting in the second state).

From Table \ref{tbl:illustrativeExample}, we see that unsurprisingly the \emph{Always East} policy results in the lowest terminal cost $c_T$ since it always seeks to move towards the goal.
The \emph{Minimum Total Belief Entropy} policy in contrast incurs a larger terminal cost $c_T$ but reduces the total entropy of the Bayesian filter estimates, but does not also minimise the smoother entropy or the total cost.
Finally, our active smoothing policy minimises the smoother entropy and the total cost, but has a larger terminal cost than the other two policies.

Our active smoothing policy offers different smoother entropy performance to the \emph{Minimum Total Belief Entropy} policy in this example since it selects actions that seek to increase the correlation between successive states, as discussed in Remark 2. 
As shown in the example control realisation in Fig.~\ref{fig:actions}, this means that our active smoothing approach more often elects to remain stationary, so as to receive measurements without changing the state.
In contrast, the \emph{Minimum Total Belief Entropy} policy moves more often since it does not directly exploit the correlation between successive states. This yields poorer state trajectory estimates.

\begin{figure}[t!]
    \centering
    \includegraphics[width = \columnwidth]{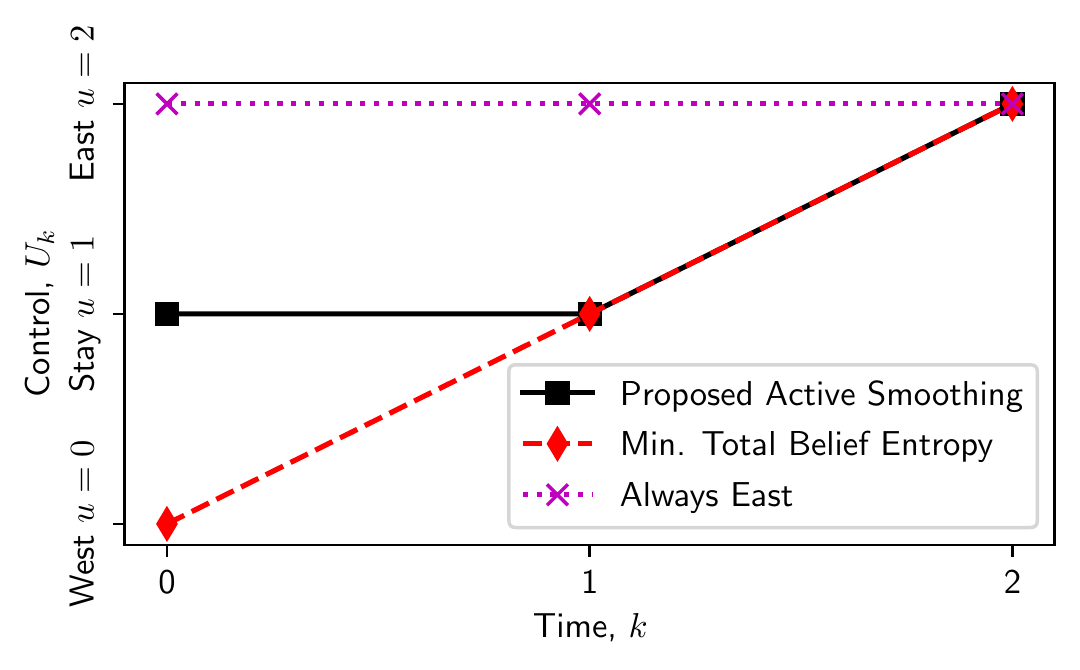}
    \caption{Example realisation of controls from each of the policies considered in our illustrative example when the agent starts in the second state.}
    \label{fig:actions}
\end{figure}

\section{Conclusion}
\label{sec:conclusion}

Active state estimation with a novel estimation performance measure directly relevant to Bayesian smoothing was proposed and investigated.
In contrast to previous active state estimation approaches, we used the joint conditional entropy of state trajectory estimates from Bayesian smoothers as a measure of estimation performance, avoiding the naive use of crude upper bounds based on summing the marginal conditional entropies of the states.
We established a novel form of the smoother entropy to show that our active smoothing problem can be reformulated as a (fully observed) Markov decision process with a value function that is concave in the belief state, enabling its (approximate) solution via piecewise-linear function approximations in conjunction with standard POMDP solvers. 
Finally, we have seen through our simulation example that control policies solving our active smoothing problem lead to improved smoother estimates compared to existing approaches based on minimising the entropy of state estimates provided by Bayesian filters.
Future work will investigate more detailed comparisons to existing approaches, and the practical application of our approach to active sensing and problems in robotics.

\bibliographystyle{IEEEtran}
\bibliography{IEEEabrv,Library}

\end{document}